\documentclass[12pt]{article}
\usepackage{amsmath,amstext,amssymb,amsfonts,amscd,verbatim,graphicx}

\newtheorem{thm}{Theorem}

\newtheorem{cor}{Corollary}
\newtheorem{prop}{Proposition}

\newtheorem{defn}{Definition}

\newtheorem{exam}{Example}

\newenvironment{proof}{{\bf Proof\,\,}}{\endproof\par}
\def \openbox{$\sqcup\llap{$\sqcap$}$}
\def \endproof{\enskip \null \nobreak \hfill \openbox \par}

\newcommand{\bz}{\mathbf 0}
\newcommand{\R}{\mathbb R}

\newcommand{\C}{\mathbb C}

\newcommand{\limto}{\rightarrow}

\begin{document}
\title{Fast and Near--Optimal Matrix Completion via Randomized Basis Pursuit}
\author{Zhisu Zhu\thanks{Institute for Computational and Mathematical Engineering, Stanford University, Stanford, CA 94305. E--mail: {\tt zhuzhisu@stanford.edu}} \and Anthony Man--Cho So\thanks{Department of Systems Engineering and Engineering Management, The Chinese University of Hong Kong, Shatin, N.~T., Hong Kong. E--mail: {\tt manchoso@se.cuhk.edu.hk}.  This research is supported by Hong Kong Research Grants Council (RGC) General Research Fund (GRF) Project No.~CUHK 2150603.} \and Yinyu Ye\thanks{Department of Management Science and Engineering, Stanford University, Stanford, CA 94305. E--mail: {\tt yinyu-ye@stanford.edu}}}
\maketitle

\begin{abstract}
Motivated by the philosophy and phenomenal success of compressed sensing, the problem of reconstructing a matrix from a sampling of its entries has attracted much attention recently.  Such a problem can be viewed as an information--theoretic variant of the well--studied matrix completion problem, and the main objective is to design an {\it efficient} algorithm that can reconstruct a matrix by inspecting only a {\it small} number of its entries.  Although this is an impossible task in general, Cand\`{e}s and co--authors have recently shown that under a so--called {\it incoherence} assumption, a rank $r$ $n\times n$ matrix can be reconstructed using semidefinite programming (SDP) after one inspects $O(nr\log^6n)$ of its entries.  In this paper we propose an alternative approach that is much more efficient and can reconstruct a larger class of matrices by inspecting a significantly smaller number of the entries.  Specifically, we first introduce a class of so--called {\it stable} matrices and show that it includes all those that satisfy the incoherence assumption.  Then, we propose a randomized basis pursuit (RBP) algorithm and show that it can reconstruct a stable rank $r$ $n\times n$ matrix after inspecting $O(nr\log n)$ of its entries.  Our sampling bound is only a logarithmic factor away from the information--theoretic limit and is essentially optimal.  Moreover, the runtime of the RBP algorithm is bounded by $O(nr^2\log n+n^2r)$, which compares very favorably with the $\Omega(n^4r^2\log^{12}n)$ runtime of the SDP--based algorithm.  Perhaps more importantly, our algorithm will provide an {\it exact} reconstruction of the input matrix in polynomial time.  By contrast, the SDP--based algorithm can only provide an {\it approximate} one in polynomial time.
\end{abstract}


\smallskip
{\bf MSC2000 Subject Classification:} 65F30, 68Q25, 68W20

\section{Introduction}
A fundamental problem that arises frequently in many disciplines is that of reconstructing a matrix with certain properties from some partial information.  Typically, such a problem is motivated by the desire to deduce global structure from a (small) number of local observations.  For instance, consider the following applications:
\begin{itemize}
   \item {\bf Covariance Estimation.} In areas such as statistics, machine learning and wireless communications, it is often of interest to find the {\it maximum likelihood estimate} of the covariance matrix $\Sigma\in\C^{m\times m}$ of a random vector $v\in\C^m$.  Such an estimate can be used to study the relationship among the variables in $v$, or to give some indication on the performance of certain systems.  Usually, extra information is available to facilitate the estimation.  For instance, we may have a number of independent samples that are drawn according to the distribution of $v$, as well as some structural constraints on $\Sigma$ (e.g., certain entries of $\Sigma^{-1}$ have prescribed values \cite{DVR08}, the matrix $\Sigma$ has a Toeplitz structure and some of its entries have prescribed values \cite{D08}, etc.).  Thus, the estimation problem becomes that of completing a partially specified matrix so that the completion satisfies the structural constraints and maximizes certain likelihood function.

   \item {\bf Graph Realization.} It is a trivial matter to see that given the coordinates of $n$ points in $\R^k$, the distance between any two points can be computed efficiently.  However, the inverse problem --- given a subset of interpoint distances, find the coordinates of points (called a {\it realization}) in $\R^k$ (where $k\ge1$ is fixed) that fit those distances --- turns out to be anything but trivial (see, e.g., \cite{Saxe79,SY06,S07}).  Such a problem arises in many different contexts, such as sensor network localization (see, e.g., \cite{BY04,SY07}) and molecular conformation (see, e.g, \cite{HW85,CH88}), and is equivalent to the problem of completing a partially specified matrix to an {\it Euclidean distance matrix} that has a certain rank (cf.~\cite{L98b,L01}).

   \item {\bf Recovering Structure from Motion.} A fundamental problem in computer vision is to reconstruct the structure of an object by analyzing its motion over time.  This problem, which is known as the {\it Structure from Motion (SfM) Problem} in the literature, can be formulated as that of finding a low--rank approximation to certain measurement matrix (see, e.g., \cite{CS04}).  However, due to the presence of occlusion or tracking failures, the measurement matrix often has missing entries.  When one takes into account such difficulties, the reconstruction problem becomes that of completing a partially specified matrix to one that has a certain rank (see, e.g., \cite{CS04}).

   \item {\bf Recommendation Systems.} Although electronic commerce has offered great convenience to customers and merchants alike, it has complicated the task of tracking and predicting customers' preferences.  To cope with this problem, various {\it recommendation systems} have been developed over the years (see, e.g., \cite{GNOT92,SKKR01,DGM08}).  Roughly speaking, those systems maintain a matrix of preferences, where the rows correspond to users and columns correspond to items.  When an user purchases or browses a subset of the items, she can specify her preferences for those items, and those preferences will then be recorded in the corresponding entries of the matrix.  Naturally, if an user has not considered a particular item, then the corresponding entry of the matrix will remain unspecified.  Now, in order to predict users' preferences for the unseen items, one will have to complete a partially specified matrix so that the completion maximizes certain performance measure (such as each individual's utility \cite{KRRT01}).
\end{itemize}
Note that in all the examples above, we are forced to take whatever information is given to us.  In particular, we cannot, for instance, specify which entries of the unknown matrix to examine.  As a result, the reconstruction problem can be ill--posed (e.g., there may not be a unique or even any solution that satisfies the given criteria).  This is indeed an important issue.  However, we shall not address it in this paper (see, e.g., \cite{J80,L98b,H01,L01} for related work).  Instead, we take a different approach and consider the information--theoretic aspects of the reconstruction problem.  Specifically, let $A\in\R^{m\times n}$ be the rank $r$ matrix that we wish to reconstruct.  For the sake of simplicity, suppose that $r$ is known.  Initially, {\it no} information about $A$ (other than its rank) is available.  However, we are allowed to inspect {\it any} entry of $A$ and inspect as many entries as we desire in order to complete the reconstruction.  Of course, if we inspect all $mn$ entries of $A$, then we will be able to reconstruct $A$ exactly.  Thus, it is natural to ask whether we can inspect only a {\it small} number of entries and still be able to reconstruct $A$ in an {\it efficient} manner. Besides being a theoretical curiosity, such a problem does arise in practical applications.  For instance, in the sensor network localization setting \cite{SY07}, the aforementioned problem is tantamount to asking which of the pairwise distances are needed in order to guarantee a successful reconstruction of the network topology.  It turns out that if the number of required pairwise distances is small, then we will be able to efficiently reconstruct the network topology by performing just a few distance measurements and solving a small semidefinite program (SDP) \cite{ZSY09}.

To get an idea of what we should aim for, let us first determine the degrees of freedom available in specifying the rank $r$ matrix $A\in\R^{m\times n}$.  This will give us a lower bound on the number of entries of $A$ we need to inspect in order to guarantee an exact reconstruction.  Towards that end, consider the singular value decomposition (SVD) $A=U\Sigma V^T$, where $U\in\R^{m\times r}$ and $V\in\R^{n\times r}$ have orthonormal columns, and $\Sigma\in\R^{r\times r}$ is a diagonal matrix.  Clearly, there are $r$ degrees of freedom in specifying $\Sigma$.  Now, observe that for $i=1,2,\ldots,r$, the $i$--th column of $U$ must be orthogonal to all of the previous $i-1$ columns, and that it must have unit length.  Thus, there are $m-i$ degrees of freedom in specifying the $i$--th column of $U$, which implies that there are $\sum_{i=1}^r(m-i)=r(2m-r-1)/2$ degrees of freedom in specifying $U$.  By the same argument, there are $\sum_{i=1}^r(n-i)=r(2n-r-1)/2$ degrees of freedom in specifying $V$.  Hence, we have:
$$ \Delta \equiv r + \frac{r(2m-r-1)}{2} + \frac{r(2n-r-1)}{2} = r(m+n-r) $$
degrees of freedom in specifying the matrix $A$.  In particular, this implies that we need to inspect at least $\Delta$ entries of $A$, for otherwise there will be infinitely many matrices that are consistent with the observations, and we will not be able to reconstruct $A$ exactly.

Now, a natural question arises whether it is possible to reconstruct $A$ by inspecting just $\Theta(\Delta)$ of its entries.  A moment of thought reveals that the answer is no, as the information that is crucial to the reconstruction of $A$ may concentrate in only a few entries.  For instance, consider the rank one $n\times n$ matrix $A=e_1e_1^T$, where $e_1=(1,0,\ldots,0)\in\R^n$.  This is a matrix with only one non--zero entry, and it is clear that if we do not inspect that entry, then there is no way we can reconstruct $A$ exactly.

From the above example, we see that our ability to reconstruct $A$ depends not only on the {\it number} of entries we inspect, but also on {\it which} entries we inspect and on the {\it structure} of $A$.  This motivates the following question: are there matrices for which exact reconstruction is possible after inspecting only $\Theta(\Delta)$ of the entries?  More generally, is there any tradeoff between the ``niceness'' of the structure of $A$ and the number of entries we need to inspect in order to reconstruct $A$?

\subsection{Related Work}
In a recent work \cite{CR09}, Cand\`{e}s and Recht studied the above questions and proposed a solution that is based on ideas from compressed sensing and convex optimization.  They first defined a notion called {\it coherence}, which can be viewed as a measure of the niceness of a matrix and is motivated by a similar notion in the compressed sensing literature \cite{CR07}.  Informally, a matrix has low coherence if the information that is crucial to its reconstruction is well--spread (cf.~the case where $A=e_1e_1^T$).  Then, they proposed the following algorithm for reconstructing any $m\times n$ matrix $A$:

\medskip
\noindent{\underline{\sc The Cand\`{e}s--Recht Algorithm}}
\begin{enumerate}
   \item Let $\Gamma$ be a uniformly random subset of $\{1,\ldots,m\}\times\{1,\ldots,n\}$ with given cardinality $|\Gamma|\ge1$.  Inspect the $(i,j)$--th entry of $A$ if $(i,j)\in\Gamma$, thus obtaining a set of values $\{A_{ij}:(i,j)\in\Gamma\}$.

   \item Output an optimal solution to the following optimization problem:
   \begin{equation} \label{eq:minnuclear}
      \begin{array}{c@{\quad}l@{\quad}l}
         \mbox{minimize} & \|X\|_* \\
         \noalign{\smallskip}
         \mbox{subject to} & X_{ij}=A_{ij} & \mbox{for } (i,j) \in \Gamma \\
         \noalign{\smallskip}
         & X \in \R^{m\times n}
      \end{array}
   \end{equation}
   Here, $\|X\|_*$ is the so--called {\it nuclear norm} of $X$ and is defined as the sum of all the singular values of $X$.
\end{enumerate}
Cand\`{e}s and Recht showed that if $A$ has low coherence, then whenever $|\Gamma|=\Omega(N^{5/4}r\log N)$, where $N=\max\{m,n\}$ and $r=\mbox{rank}(A)$, the solution to problem (\ref{eq:minnuclear}) will be unique and equal to $A$ with high probability.  In other words, by inspecting $O(N^{5/4}r\log N)$ randomly chosen entries of $A$ and then solving the optimization problem (\ref{eq:minnuclear}), one can reconstruct $A$ exactly with high probability.  Note that problem (\ref{eq:minnuclear}) can be formulated as a SDP; see, e.g., \cite[Chapter 5]{F02}.  As such, it can be solved to any desired accuracy in polynomial time (see, e.g., \cite{GLS93,T01}).  However, if one uses standard SDP solvers, then the runtime of the Cand\`{e}s--Recht algorithm is at least on the order of $\max\{N^{9/2}r^2\log^2N,N^{15/4}r^3\log^3N\}$ (see, e.g., \cite{T01,LV09}), which severely limits its use in practice.  Although specialized algorithms are being developed to solve the SDP associated with problem (\ref{eq:minnuclear}) more efficiently (see, e.g., \cite{CCS08,HH09,LV09,ZSY09}), they either do not have any theoretical time bound, or their runtimes can still be prohibitively high when $N$ is large (at least on the order of $N^{9/2}r^2\log^2N$).

Subsequent to the work of Cand\`{e}s and Recht, improvements have been made by various researchers on both the sampling and runtime bounds for the problem.  In \cite{KMO09}, Keshavan et al.~proposed a reconstruction algorithm that is based on the SVD and a certain manifold optimization procedure.  They showed that if the input matrix $A$ has low coherence {\it and} low rank, then by sampling $|\Gamma|=\Omega(Nr\max\{\log N,r\})$ entries of $A$ uniformly at random, their algorithm will produce a sequence of iterates that converges to $A$ with high probability.  Note that the sampling complexity of Keshavan et al.'s algorithm is just a polylogarithmic factor away from the information--theoretic minimum $\Delta$ and hence is almost optimal.  However, their result applies only when the rank of $A$ is bounded above by $N^{1/2}$, and the ratio between the largest and smallest singular values is bounded.  Moreover, there is no theoretical time bound for their algorithm.  Around the same time, Cand\`{e}s and Tao \cite{CT09} refined the analysis in \cite{CR09} and showed that the sampling complexity of the Cand\`{e}s--Recht algorithm can be reduced to $|\Gamma|=\Omega(Nr\log^6N)$ when the input matrix $A$ has low coherence (but not necessarily low rank).  Again, this is just a polylogarithmic factor away from the information--theoretic minimum $\Delta$.  However, the runtime of the algorithm remains high (at least on the order of $N^4r^2\log^{12}N$).

\subsection{Our Contribution}
From the above discussion, we see that it is desirable to design a reconstruction algorithm that can work for a large class of matrices and yet still has low sampling and computational complexities.  In this paper we make a step towards that goal.  Specifically, our contribution is twofold.  First, we introduce the notion of {\it $k$--stability}, which again can be viewed as a measure of the niceness of a matrix.  Roughly speaking, an $m\times n$ rank $r$ matrix is $k$--stable if {\it every} $m\times(n-k)$ sub--matrix of $A$ has rank $r$, but {\it some} $m\times(n-k-1)$ sub--matrix of $A$ has rank $r-1$.  Intuitively, if a low--rank matrix has high stability (i.e.~when $k$ is large), then the information that is crucial to its reconstruction is present in many small subsets of the columns, and hence the matrix should be more amenable to exact reconstruction.  As it turns out, the notion of $k$--stability is related to the so--called {\it Maximum Distance Separable (MDS) codes} in coding theory \cite[Chapter 11]{MS77}.  Moreover, from the above informal definition, we see that $k$--stability is a {\it combinatorial} property of matrices, which should be contrasted with the more {\it analytic} nature of the notion of coherence as defined in \cite{CR09}.  Nevertheless, there is a strong connection between those two notions.  More precisely, we show that if a matrix has low coherence, then it must have high stability.  Such a connection opens up the possibility of comparing our results to those in \cite{CR09,KMO09,CT09}.

Secondly, we propose a randomized basis pursuit (RBP) algorithm for the reconstruction problem.  Our algorithm differs from that of Cand\`{e}s and Recht \cite{CR09} and Keshavan et al.~\cite{KMO09} in two major aspects:
\begin{enumerate}
   \item We do not sample the matrix {\it entries} in a uniform fashion.  Instead, we sample the {\it columns} (or {\it rows}) of the matrix uniformly.  We note that such a sampling strategy is reminiscent of that used for constructing low--rank approximations to a given matrix \cite{FKV04,DKM06,RV07}.  However, there is one crucial difference, namely, our sampling strategy does not require any knowledge of the input matrix.  By contrast, the strategy used in \cite{FKV04,DKM06,RV07} assumes that the norm of each column of the input matrix is known.

   \item Our algorithm does not involve any optimization procedure and will produce an {\it exact} solution in polynomial time.  This should be contrasted with the SDP--based algorithm of Cand\`{e}s and Recht, which can only return an {\it approximate} solution in polynomial time (see \cite{PK97} for discussions on the complexity of solving SDPs); and with the spectral method of Keshavan et al., which is known to converge to an exact solution but has no theoretical time bound.
\end{enumerate} 
Regarding the performance of our algorithm, we show that if the input matrix $A$ has high stability (in particular, this includes the case where $A$ has low coherence), then by sampling $O(Nr\log N)$ entries of $A$ using our column sampling procedure, we can reconstruct $A$ exactly with high probability.  Furthermore, we show that the runtime of our algorithm is bounded above by $O(Nr^2\log N+N^2r)$.  Thus, on both the sampling and computational complexities, our bounds yield substantial improvement over those in \cite{KMO09,CT09}.  Moreover, our sampling bound is essentially optimal, as the extra $\log N$ factor can be attributed to the coupon collecting phenomenon \cite[Chapter 3]{MR95} (see \cite{CR09,KMO09,CT09} for related discussions).

\subsection{Outline of the Paper}
In Section \ref{sec:stablematrix} we first introduce the notion of a $k$--stable matrix and derive some of its properties.  Then, we study the relationship between the notion of $k$--stability and the notion of coherence defined in \cite{CR09}.  Afterwards, we study some constructions of $k$--stable matrices and show that they are in fact quite ubiquitous.  In Section \ref{sec:RBP} we propose a randomized basis pursuit (RBP) algorithm for the matrix reconstruction problem and analyze its sampling and computational complexities.  Although the RBP algorithm assumes that the rank of the input matrix is known, we show how such an assumption can be removed in Section \ref{subsec:rank-free}.  Finally, we summarize our results and discuss some possible future directions in Section \ref{sec:concl}.

\section{The Class of $k$--Stable Matrices} \label{sec:stablematrix}
As mentioned in the Introduction, our ability to reconstruct a matrix depends in part on its structure.  In this paper we shall focus on the class of $k$--stable matrices, which is defined as follows:
\begin{defn}
A rank $r$ matrix $A\in\R^{m \times n}$ is said to be \emph{$k$--stable} for some $k\in\{0,1,\ldots,n-r\}$ if
\begin{enumerate}
   \item every $m \times (n-k)$ sub--matrix of $A$ has rank $r$; and

   \item there exists an $m \times (n-k-1)$ sub--matrix of $A$ with rank equal to $r-1$.
\end{enumerate}
In other words, the rank of a $k$--stable matrix $A$ remains unchanged under the removal of any of its $k$ columns.  We use $\mathcal{M}^{m \times n}(k,r)\subset\R^{m\times n}$ to denote the set of all $k$--stable rank $r$ $m\times n$ matrices, and use $\mathcal{M}^{m \times n}(k)\subset\R^{m\times n}$ to denote the set of all $k$--stable $m\times n$ matrices.
\end{defn}
Note that the notion of $k$--stability is defined with respect to the columns of a matrix.  Of course, we may also define it with respect to the rows.  However, unlike the notions of row rank and column rank --- which are equivalent --- a column $k$--stable matrix may not be row $k$--stable.  Unless otherwise stated, we shall refer to a {\it column} $k$--stable matrix simply as a $k$--stable matrix in the sequel.

As we shall see, the notion of $k$--stability has many nice properties.  For instance, it generalizes the notions of coherence defined in \cite{CR09,CT09}.  Moreover, a matrix with high stability (i.e.~$k=\Theta(n)$) can be reconstructed by a simple and efficient randomized algorithm with high probability.  Before we give the details of these results, let us first take a look at some (deterministic) constructions of $k$--stable matrices.
\begin{exam} \label{exam:first-row}
   Let $a=(a_1,\ldots,a_n)\in\R^n$ be any vector with no zero component.  Consider the $m\times n$ matrix $A$ whose first row is equal to $a^T$ and all other entries are zeroes.  Then, $A$ is an $(n-1)$--stable rank one matrix.
\end{exam}
\begin{exam}
   Let $n\ge1$ be an odd integer.  Let $e\in\R^n$ be the vector of all ones, and let
   $$ u=\left(-\frac{n-1}{2},-\frac{n-1}{2}+1,-\frac{n-1}{2}+2,\ldots,\frac{n-1}{2}\right)\in\R^n $$
   Consider the $n\times n$ matrix $A$ whose first row is equal to $e^T$ and the $i$--th row is equal to $u^T$, where $i=2,\ldots,n$.  It is then easy to verify that $A$ is an $(n-2)$--stable rank two matrix.
\end{exam}
\begin{exam}
   Let $m,n$ be integers with $n\ge m\ge1$.  Suppose that $u=(u_1,\ldots,u_m)\in\R^m$ and $v=(v_1,\ldots,v_n)\in\R^n$ are given.  Consider the $m\times n$ matrix $A$ defined by $A_{ij}=(u_i+v_j)^{-1}$, where $1\le i\le m$ and $1\le j\le n$.  The matrix $A$ is known as a \emph{Cauchy} matrix.  It is well--known that if the $u_i$'s are distinct, the $v_j$'s are distinct, and $u_i+v_j\not=0$ for all $1\le i\le m$ and $1\le j\le n$, then every square sub--matrix of $A$ is non--singular.  In particular, this implies that $A$ is an $(n-m)$--stable rank $m$ matrix.
\end{exam}

\subsection{Relation to the Notion of Coherence}
In all previous work on the matrix reconstruction problem \cite{CR09,KMO09,CT09}, the notion of coherence is used to measure the niceness of a matrix.  This immediately raises the question of whether $k$--stability and coherence are comparable notions.  It turns out that the former can be viewed as a generalization of the latter.  Before we formalize this statement, let us first recall the definition of coherence \cite{CR09}:
\begin{defn}
Let $U\subset\R^n$ be a subspace of dimension $r\ge1$, and let $P_U$ be the
orthogonal projection onto $U$. Then, the \emph{coherence} of $U$ is defined as:
$$ \mu(U)\equiv \frac{n}{r} \max_{1\leq i \leq n} \|P_Ue_i\|_2^2 $$
where $e_i\in\R^n$ is the $i$--th standard basis vector, for $i=1,\ldots,n$.
\end{defn}
A simple consequence of this definition is the following:
\begin{prop} \label{prop:orth-coherence}
   Let $U\in\R^{n\times r}$ be a matrix with orthonormal columns.  When viewed as a subspace of $\R^n$ (i.e.~the subspace spanned by the columns of $U$), the coherence of $U$ is given by:
$$ \mu(U) = \frac{n}{r}\max_{1\le i\le n} \|u_i\|_2^2 $$
where $u_i$ is the $i$--th row of $U$, for $i=1,\ldots,n$.
\end{prop}
\begin{proof}
   Let $v_1,\ldots,v_r\in\R^n$ be the columns of $U$.  Since $U$ has orthonormal columns, we have:
   $$ P_Ue_i = \sum_{j=1}^r (e_i^Tv_j)v_j \quad\mbox{for }i=1,\ldots,n $$
   whence:
   $$ \|P_Ue_i\|_2^2 = \sum_{j=1}^r (e_i^Tv_j)^2 = \|u_i\|_2^2 $$
as desired.
\end{proof}

\medskip
\noindent The following invariance property of $k$--stable matrices will be useful for establishing the relationship between $k$--stability and coherence:
\begin{prop} \label{prop:invar}
   Let $A\in\R^{m\times n}$ be an arbitrary matrix, and let $U\in\R^{p\times m}$ be a matrix with linearly independent columns (in particular, we have $p\ge m$).  Then, for any $k\in\{0,1,\ldots,n-r\}$, $A$ is a $k$--stable rank $r$ matrix iff $UA$ is so.
\end{prop}
\begin{proof}
   Let $a_1,\ldots,a_n\in\R^m$ be the columns of $A$.  Then, the columns of $UA$ are given by $Ua_1,\ldots,Ua_n\in\R^p$.  Now, for any $l=1,\ldots,n$, the number of linearly independent vectors in the collection $\{a_{i_1},\ldots,a_{i_l}\}$ is the same as that in the collection $\{Ua_{i_1},\ldots,Ua_{i_l}\}$, since $U$ has full column rank.  This completes the proof.
\end{proof}

\medskip
\noindent We are now ready to state our first main result:
\begin{thm} \label{thm:stab-coher}
   Let $A\in\R^{m\times n}$ be a rank $r$ matrix whose SVD is given by $A= U\Sigma V^T$, where $U\in\R^{m\times r}$, $V\in\R^{n\times r}$ and $\Sigma\in\R^{r\times r}$.  For any non--negative integer $k\leq n-r$, if the coherence of $V$ satisfies $\mu(V)\leq \mu_0$ for some $\mu_0\in(0,\frac{n}{kr})$, then $A$ is column $s$--stable for some $s\geq k$.  Similarly, for any non--negative integer $k'\le m-r$, if $\mu(U)\leq \mu_0$ for some $\mu_0\in(0,\frac{m}{k'r})$, then $A$ is row $s$--stable for some $s\geq k'$.
\end{thm}
\begin{proof}
By Proposition \ref{prop:invar}, it suffices to show that $V^T$ is a column $k$--stable rank $r$ matrix, since $U\Sigma\in\R^{m\times r}$ is a matrix with linearly independent columns.  Now, consider the following cases:

\medskip
\noindent{\it\underline{Case 1}}: $r=1$

\medskip
\noindent Let $V=(v_1,\ldots,v_n)\in\R^n$.  Then, by Proposition \ref{prop:orth-coherence} and the fact that $\mu(V)\le\mu_0$, we have:
$$ \mu(V)=n\max_{1\le i\le n} v_i^2 \le \mu_0 < \frac{n}{k} $$
It follows that $v_i^2<1/k$ for $i=1,\ldots,n$.  Since $\sum_{i=1}^nv_i^2=1$, we conclude that $V$ must have at least $k+1$ non--zero entries.  It follows that $V^T$ is a column $s$--stable rank one matrix for some $s\geq k$, since the removal of any $k$ columns from $V^T$ does not change its rank.

\medskip
\noindent{\it\underline{Case 2}}: $r\ge2$

\medskip
\noindent Suppose that $V^T$ is only an $l$--stable rank $r$ matrix for some $0\le l\le k-1$.  Then, by definition, there exist $l+1$ columns whose removal will result in a rank $r-1$ sub--matrix of $V^T$.  Without loss of generality, suppose that those $l+1$ columns are the last $l+1$ columns of $V^T$.  Then, we may write $V^T=[RQ \,\,\, N]$, where $R \in \R^{r \times (r-1)}, Q \in \R^{(r-1) \times (n-l-1)}, N \in \R^{r \times (l+1)}$, and $Q$ has orthonormal rows.  Since $V^T$ has orthonormal rows, we have:
$$ I_r=V^T V=RR^T+NN^T = [R\,\,\,N][R\,\,\,N]^T $$
which means that the matrix $[R\,\,\,N]\in\R^{r\times(r+l)}$ also has orthonormal rows.  In particular, we have $\|R\|_F^2\leq r-1$ (here, $\|\cdot \|_F$ is the Frobenius norm).  Moreover, since $\|[R\,\,\,N]\|_F^2=r$, we have:
\begin{equation} \label{eq:f-norm}
   \|N\|_F^2= \|[R\,\,\,N]\|_F^2 - \|R\|_F^2 \geq 1
\end{equation}
On the other hand, observe that:
$$ \|N\|_F^2 \le \frac{lr\mu_0}{n} < \frac{l}{k} < 1 $$
which contradicts (\ref{eq:f-norm}).  Thus, we conclude that $V^T$ (and hence $A$) is a column $s$--stable rank $r$ matrix for some $s\ge k$.

The statement about the row stability of $A$ can be established by considering $A^T=V\Sigma U^T$ and following the above argument.
\end{proof}

\medskip
\noindent One of the consequences of Theorem \ref{thm:stab-coher} is that if both the factors $U$ and $V$ in the SVD of $A$ have small coherence relative to $\min\{m,n\}/r$ (which is the case of interest in the work \cite{CR09,KMO09,CT09}), then $A$ has high row and column stability.  Now, a natural question arises whether the converse also holds.  Curiously, as the following proposition shows, the answer is no.
\begin{prop} \label{prop:stab-not-coher}
   Let $k\in\{1,\ldots,n-1\}$ be arbitrary.  Then, there exist $n\times n$ rank one matrices $A$ that are both row and column $k$--stable, and yet the corresponding SVDs $A=\sigma uv^T$ satisfy $\min\{\mu(u),\mu(v)\}=\Theta(n)$.
\end{prop}
\begin{proof}
   Let $\epsilon\in(0,1/2)$ be fixed.  Define $u\in\R^n$ as follows:
   $$ u_i = \left\{
   \begin{array}{c@{\quad}l}
      \sqrt{1-\epsilon} & \mbox{if } i=1 \\
      \noalign{\smallskip}
      \sqrt{\epsilon/k} & \mbox{if } 2 \le i\le k+1 \\
      \noalign{\smallskip}
      0 & \mbox{otherwise}
   \end{array}
   \right.
   $$
By construction, we have $\|u\|_2^2=1$.  Now, consider the rank one matrix $A=uu^T$.  Since $u^T$ has $k+1$ non--zero entries, it is column $k$--stable.  Thus, by Proposition \ref{prop:invar} and the symmetry of $A$, we conclude that $A$ is both row and column $k$--stable.  On the other hand, using Proposition \ref{prop:orth-coherence}, we compute:
$$ \mu(u) = n\max_{1\le i\le n} u_i^2 = (1-\epsilon)n $$
This completes the proof.
\end{proof}

\medskip
\noindent As can be seen in the proof of Proposition \ref{prop:stab-not-coher}, the coherence of a matrix can be very sensitive to the actual values in its entries.  This can be partly attributed to the fact that coherence is an analytic notion.  By contrast, the notion of $k$--stability is more combinatorial in nature and hence is not as sensitive to those values.

Theorem \ref{thm:stab-coher} and Proposition \ref{prop:stab-not-coher} together show that the notion of $k$--stability can be regarded as a generalization of the notions of coherence defined in \cite{CR09,CT09}.  In particular, various constructions of low--coherence matrices proposed in \cite{CR09,CT09} can be transferred to the high--stable case.  However, it would be nice to have some more direct constructions of high--stable matrices.  In the next section, we will show that matrices with high stability are actually quite ubiquitous.

\subsection{Ubiquity of $k$--Stable Matrices}
Let $A\in\R^{r\times n}$ be a matrix with full row rank (in particular, we have $r\le n$).  Then, it is clear that the maximum stability of $A$ is $n-r$, and that the maximum can be attained.  It turns out that such a situation is typical.  More precisely, we have the following:
\begin{thm} \label{thm:stab-typ}
   Let $r,n$ be integers with $n\ge r\ge1$.  Then, the set $\mathcal{S}\equiv\R^{r\times n}\backslash\mathcal{M}^{r\times n}(n-r,r)$ has Lebesgue measure zero when considered as a subset of $\R^{rn}$.
\end{thm}
The proof of Theorem \ref{thm:stab-typ} relies on the following well--known result:
\begin{prop} \label{prop:poly-soln-meas}
   Let $f:\R^l\limto\R$ be a polynomial function that is not identically equal to zero.  Then, the solution set:
$$ f^{-1}(0)\equiv\{x\in\R^l:f(x)=0\} $$
has Lebesgue measure zero.
\end{prop}
A proof of Proposition \ref{prop:poly-soln-meas} can be found in \cite{O73}.

\medskip
\noindent{\bf Proof of Theorem \ref{thm:stab-typ}}{\quad}Suppose that $A\in\R^{r\times n}$ is not $(n-r)$--stable.  Then, one of the $r\times r$ sub--matrices of $A$ must be singular, or equivalently, has determinant zero.  Since the determinant of a square matrix is a polynomial function of its entries, and since there are only finitely many $r\times r$ sub--matrices of $A$, it follows from Proposition \ref{prop:poly-soln-meas} that $\mathcal{S}$ has Lebesgue measure zero. \endproof

\medskip
\noindent Thus, by taking a generic $r\times n$ matrix $R$ and an arbitrary $m\times r$ matrix $Q$ whose columns are linearly independent, we may conclude from Proposition \ref{prop:invar} and Theorem \ref{thm:stab-typ} that the $m\times n$ matrix $A=QR$ has rank $r$ and is column $(n-r)$--stable.

In \cite{CR09} the authors considered an alternative construction of rank $r$ matrices using the so--called {\it random orthogonal model}.  In that model, one constructs an $m\times n$ matrix $A$ via $A=U\Sigma V^T$, where $V\in\R^{n\times n}$ is a random orthogonal matrix drawn according to the Haar measure on the orthogonal group $\mathcal{O}(n)$, $U\in\R^{m\times m}$ is an arbitrary orthogonal matrix, and $\Sigma\in\R^{m\times n}$ is an arbitrary matrix with the partition:
$$ \Sigma=\left[\begin{array}{cc} \Sigma_r & \bz \\ \noalign{\smallskip} \bz & \bz \end{array} \right] \,:\, \Sigma_r\in\R^{r\times r} \mbox{ diagonal},\,\,(\Sigma_r)_{ii}\not=0\,\mbox{ for }i=1,\ldots,r $$
By construction, the matrix $A$ has rank $r$.  Now, we claim that $A$ is column $(n-r)$--stable with probability one (with respect to the Haar measure on $\mathcal{O}(n)$).  To see this, observe that $A=U_r\Sigma_rV_r^T$, where:
$$ U=\left[\begin{array}{cc} U_r & \bar{U}_r\end{array}\right],\quad V=\left[\begin{array}{cc} V_r & \bar{V}_r\end{array}\right] $$
and $U_r\in\R^{m\times r}$, $V_r\in\R^{n\times r}$.  Since $U_r\Sigma_r$ has linearly independent columns, by Proposition \ref{prop:invar}, it suffices to show that $V_r^T$ is column $(n-r)$--stable with probability one.  Towards that end, it suffices to show that with probability one, every $r\times r$ sub--matrix of $V_r^T$ has non--zero determinant.  It turns out that the last statement is well--known; see, e.g., \cite[Lemma 2.2]{KW85}.  Thus, we have proven the following:
\begin{thm} \label{thm:rand-orth-stab}
   Let $A$ be a rank $r$ $m\times n$ matrix generated according to the random orthogonal model.  Then, $A$ is column $(n-r)$--stable with probability one (with respect to the Haar measure on $\mathcal{O}(n)$).
\end{thm}
In \cite{CR09} it is shown that the coherence of a rank $r$ $n\times n$ matrix generated according to the random orthogonal model is bounded by $O(\bar{r}/r)$ with probability $1-o(1)$, where $\bar{r}=\max\{r,\log n\}$.  Using Theorem \ref{thm:stab-coher}, we see that such a matrix is column $(n/\bar{r})$--stable with probability $1-o(1)$.  This should be contrasted with the conclusion of Theorem \ref{thm:rand-orth-stab}, which is much stronger.

\section{The Randomized Basis Pursuit (RBP) Algorithm} \label{sec:RBP}
Let us now consider the algorithmic aspects of matrix reconstruction, particularly those that are related to the reconstruction of low--rank high--stable matrices.  As briefly discussed in the Introduction, if a reconstruction algorithm can only inspect a small number of entries, then it should somehow inspect those that contain the most information.  Of course, since there is no a priori information on the input matrix, every algorithm must at some point make a guess at which entries are important.  Currently, the best algorithms for the reconstruction problem all pursue an entry--wise uniform sampling strategy \cite{CR09,KMO09,CT09}.  Specifically, they all begin by sampling a uniformly random subset of the entries and inspecting the values in those entries.  Such a strategy will certainly perform well when the information that is crucial to the reconstruction is well--spread, but could also fail miserably when those information is highly concentrated.  As an illustration, consider the rank one $m\times n$ matrix $A$ from Example \ref{exam:first-row}, which has the form:
\begin{equation} \label{eq:exam1}
   A = \left[\begin{array}{cc} a^T \\ \noalign{\smallskip} \bz \end{array}\right]
\end{equation}
where $a\in\R^n$ has no zero component.  Clearly, there is no hope of reconstructing $A$ if we do not inspect all the entries in its first row.  However, if the entry--wise uniform sampling strategy is used, then the probability that $l$ randomly sampled entries of $A$ will include all the entries in the first row is bounded above by:
$$ \frac{\left(
\begin{array}{c}
 mn-n \\
  l-n \\
   \end{array}
   \right)
}{\left(
 \begin{array}{c}
 mn \\
   l \\
   \end{array}
   \right)}=\frac{l(l-1)\cdots (l-n+1)}{mn(mn-1)\cdots (mn-n+1)}\leq \left(\frac{l}{mn}\right)^n
$$
In particular, {\it no} algorithm that uses the entry--wise uniform sampling strategy will be able to reconstruct $A$ with probability larger than $e^{-1}$ even after sampling $l=mn-m=\Theta(mn)$ of its entries!

The above example shows that the entry--wise uniform sampling strategy may miss the critical structure of a matrix if that structure is localized.  On the other hand, observe that the matrix $A$ in the above example can be exactly reconstructed once we inspect its first row and {\it any} of its columns.  In general, one may think of a low--rank matrix as being largely determined by a small number of its rows and columns.  Such an intuition motivates the following matrix reconstruction algorithm.  Note that the algorithm requires the knowledge of the rank of the input matrix.  However, as we shall see in Section \ref{subsec:rank-free}, such an assumption can be removed if we can bound the stability of the input matrix.

\medskip
\noindent{\underline{\sc Randomized Basis Pursuit (RBP) Algorithm}}

\smallskip
\noindent{\,\,\,}\underline{Input}: A rank $r$ $m\times n$ matrix $A$, where $r$ is known.
\begin{enumerate}
   \item \underline{Initialization}: Initialize $S\leftarrow\emptyset$ and $T\leftarrow\{1,\ldots,n\}$.  The set $S$ will be used to store the column indices that correspond to the recovered basis columns of $A$.

   \item \underline{Basis Pursuit Step}:
   \begin{enumerate}
      \item If $T=\emptyset$, then stop.  All the columns of $A$ have been examined, and hence $A$ can be reconstructed directly.

      \item Otherwise, let $j$ be drawn from $T$ uniformly at random, and let $u_j\in\R^m$ be the corresponding column of $A$.  Examine all the entries in $u_j$.  If $u_j$ is spanned by the columns whose indices belong to $S$, then repeat Step 2b.  Otherwise, update:
      $$ S\leftarrow S\cup\{j\},\quad T\leftarrow T\backslash\{j\} $$
      since $u_j$ is a new basis column.  Now, if $|S|=r$, then proceed to Step 3.  Otherwise, repeat Step 2.
   \end{enumerate}

   \item \underline{Row Identification}: Let $A_S$ be the $m\times r$ sub--matrix of $A$ whose columns are those indexed by $S$.  Find $r$ linearly independent rows in $A_S$.  Let $\bar{S}$ be the corresponding set of row indices, and let $A_{\bar{S},S}$ be the corresponding $r\times r$ matrix.

   \item \underline{Reconstruction}: Examine all the entries in the $i$--th row of $A$ for all $i\in\bar{S}$.  Now, the $j$--th column of $A$ (where $j\not\in S$) can be expressed as a linear combination of the basis columns indexed by $S$, where the coefficients are obtained by expressing the vector $(a_{ij})_{i\in\bar{S}}\in\R^{|\bar{S}|}$ as a linear combination of the columns of $A_{\bar{S},S}$.
\end{enumerate}
It is not hard to show that when the above algorithm terminates, it will produce an exact reconstruction of the input matrix.  To illustrate the flow of the algorithm, let us consider again the rank one matrix $A$ from Example \ref{exam:first-row} (see (\ref{eq:exam1})).  Since the first row of $A$ has no zero component, any column selected in Step 2b of the algorithm can be the basis column.  Now, suppose that $j$ is the index of the selected column.  After inspecting all the entries in the $j$--th column, the algorithm will identify the $1\times 1$ sub--matrix $A_{1j}$ in Step 3, since $A_{1j}$ is the only non--zero entry in the $j$--th column.  Consequently, the algorithm will examine all the entries in the first row of $A$ in Step 4, thus obtaining all the information that is necessary for the reconstruction of $A$.  Note that in this example, the total number of entries inspected by the algorithm is $m+n-1$, which is exactly equal to the information--theoretic minimum.

From the description of the RBP algorithm, we see that if the input matrix is of low rank but has many candidate basis columns, then the basis pursuit step will terminate sooner, and hence the number of entries inspected by the algorithm will also be lower.  This is indeed the case when the input matrix has high stability (recall that the matrix from Example \ref{exam:first-row} is $(n-1)$--stable).  Before we proceed with a formal analysis, let us remark that some additional speed up of the above algorithm is possible.  For instance, in Step 2b, once we determine that a column lies in the span of those indexed by $S$, then we do not need to consider it anymore and hence its index can be removed from $T$.  However, in order to facilitate the analysis, we shall focus on the version presented above.

\subsection{Sampling Complexity of the RBP Algorithm}
In this section we analyze the sampling complexity of the RBP algorithm.  Specifically, our goal is to prove the following:
\begin{thm} \label{thm:RBP-sample}
   Suppose that the input rank $r$ $m\times n$ matrix $A$ to the RBP algorithm is $k$--stable for some $k\in\{0,1,\ldots,n-r\}$, i.e.~$A\in\mathcal{M}^{m\times n}(k,r)$.  Let $\delta\in(0,1)$ be given.  Then, with probability at least $1-r\delta$, the RBP algorithm will terminate with an exact reconstruction of $A$, and the total number of entries inspected by the algorithm is bounded above by $nr+(k+1)^{-1}mnr(1+\ln(1/\delta))$.
\end{thm}
The following simple estimate will be used in the proof of Theorem \ref{thm:RBP-sample}:
\begin{prop} \label{prop:geom-ld}
   Let $X$ be a geometric random variable with parameter $p\in(0,1)$.  Then, for any $\delta>0$, we have:
   $$ \Pr\left(X>\frac{1+\delta}{p}\right) \le e^{-\delta} $$
\end{prop}
\begin{proof}
   We compute:
   \begin{eqnarray*}
      \Pr\left(X>\frac{1+\delta}{p}\right) &\le& \sum_{j=\lceil(1+\delta)/p\rceil}^{\infty} p(1-p)^{j-1} \\
      \noalign{\medskip}
      &=& p\cdot(1-p)^{\lceil(1+\delta)/p\rceil-1}\cdot\sum_{j=0}^{\infty}(1-p)^j \\
      \noalign{\medskip}
      &\le& (1-p)^{\delta/p} \\
      \noalign{\medskip}
      &\le& e^{-\delta}
   \end{eqnarray*}
This completes the proof.
\end{proof}

\medskip
\noindent{\bf Proof of Theorem \ref{thm:RBP-sample}}{\quad} Observe that Step 2 of the RBP algorithm is the only place where randomization is used, and that once Step 2 is completed, the algorithm will always terminate with an exact reconstruction of $A$.  Thus, it suffices to obtain a high probability bound on the number of times Step 2b is being executed throughout the entire course of the algorithm.  Towards that end, let us divide the execution of Step 2 into epochs, where the $i$--th epoch (for $i=0,1,\ldots,r-1$) begins at the iteration where $|S|=i$ for the first time and ends at the iteration just before the one where $|S|=i+1$.  Let $p_i$ be the probability that the column selected in an iteration of the $i$--th epoch is a basis column.  Note that $p_i$ is a random variable that depends on which $i$ basis columns are selected in the previous $i$ epochs.  However, since the input matrix is assumed to be $k$--stable, we have:
$$ p_i \ge \frac{k+1}{n-i} \qquad\mbox{for } i=0,1,\ldots,r-1 $$
Now, let $Y_i$ be the number of times Step 2b is being executed in the $i$--th epoch.  Then, $Y_i$ is a geometric random variable with parameter $p_i$, and the number of times Step 2b is being executed throughout the entire course of the algorithm is given by:
$$ Y = \sum_{i=0}^{r-1}Y_i $$
By Proposition \ref{prop:geom-ld}, we have:
$$ \Pr\left(Y>\sum_{i=0}^{r-1}\frac{1+\ln(1/\delta)}{p_i}\right) \le \sum_{i=0}^{r-1} \Pr\left(Y_i > \frac{1+\ln(1/\delta)}{p_i}\right) \le r\delta $$
It follows that with probability at least $1-r\delta$, the total number of times Step 2b is being executed is bounded above by:
   \begin{eqnarray*}
      \sum_{i=0}^{r-1}\frac{1+\ln(1/\delta)}{p_i} &\le& \left(1+\ln\frac{1}{\delta}\right)\sum_{i=0}^{r-1}\frac{n-i}{k+1} \\
      \noalign{\medskip}
      &=&  \frac{r(2n-r+1)}{2(k+1)}\left(1+\ln\frac{1}{\delta}\right) \\
      \noalign{\medskip}
      &\le& \frac{nr}{k+1}\left(1+\ln\frac{1}{\delta}\right)
   \end{eqnarray*}
Note that the above quantity is also an upper bound on the number of {\it distinct} columns examined by the algorithm.  Thus, we see that with probability at least $1-r\delta$, the total number of entries inspected by the algorithm is bounded above by:
$$ nr+\frac{mnr}{k+1}\left(1+\ln\frac{1}{\delta}\right) $$
and the proof is completed. \endproof

\medskip
\noindent Upon combining the results of Theorem \ref{thm:rand-orth-stab} and Theorem \ref{thm:RBP-sample}, we obtain the following corollary, which significantly improves the result in \cite{CT09}:
\begin{cor}
   Let $A$ be a rank $r$ $m\times n$ matrix generated according to the random orthogonal model.  Then, with probability at least $1-O(n^{-3})$, the RBP algorithm will terminate with an exact reconstruction of $A$, and the total number of entries inspected by the algorithm is bounded above by $O(n+m\log n)$ when $r=O(1)$, and by $O(n\log n+m\log^2n)$ when $r=O(\log n)$.
\end{cor}

\subsection{Implementation and Complexity Analysis of the RBP Algorithm} 
In this section we discuss some of the implementation details of the RBP algorithm and analyze its computational complexity.  Clearly, the initialization step can be done using $O(n)$ operations.  For the basis pursuit step, we need to determine whether a newly selected column is in the span of the current basis columns (i.e.~those indexed by $S$).  This can be achieved via a Gram--Schmidt type process.  Specifically, we maintain a set $\mathcal{U}$ of orthonormal vectors with the following property: during the $i$--th epoch (where $i=0,1,\ldots,r-1$), the set $\mathcal{U}$ will contain $i$ orthonormal vectors $w_1,\ldots,w_i\in\R^m$, whose span is equal to that of the columns indexed by $S$.  Now, suppose that the algorithm selects the column $v\in\R^m$.  To test whether $v\in\mbox{span}\{w_1,\ldots,w_i\}$, we first compute:
$$ \Pi_i(v) \equiv \sum_{l=1}^i (w_l^Tv)w_l \in \R^m $$
and then check whether $\Pi_i(v)=v$ (we set $\Pi_0(v)=\bz$).  If this is the case, then we have $v\in\mbox{span}\{w_1,\ldots,w_i\}$, whence we can proceed to select another column.  Otherwise, $v$ is a new basis column.  Thus, we add its index to the set $S$ and add the unit vector $(v-\Pi_i(v))/\|v-\Pi_i(v)\|_2$ to the set $\mathcal{U}$ before continuing to the next instruction.

To determine the time needed by the basis pursuit step, observe that during the $i$--th epoch, each selected column requires $O(im)$ operations.  Since $i\le r-1$, by Theorem \ref{thm:RBP-sample}, we conclude that with probability $1-r\delta$, the total number of operations executed in the basis pursuit step is bounded by $O((k+1)^{-1}mnr^2\log(1/\delta))$.

In the row identification step, we need to find $r$ linearly independent rows in the $m\times r$ matrix $A_S$.  This can be achieved by a Gram--Schmidt type process similar to the one outlined above, and the total number of operations required is bounded by $O(mr^2)$.

Finally, in order to carry out the reconstruction step, we can first compute the inverse of the non--singular $r\times r$ matrix $A_{\bar{S},S}$ using $O(r^3)$ operations.  Then, for each $j\not\in S$, we can express the vector $(a_{ij})_{i\in\bar{S}}\in\R^{|\bar{S}|}$ as a linear combination of the columns of $A_{\bar{S},S}$ using $O(r^2)$ operations.  Afterwards, the $j$--th column can be reconstructed using $O(mr)$ operations.  Since we need to reconstruct at most $n-1$ columns, it follows that the total number of operations required in the
reconstruction step is bounded by $O(r^3+nr^2+mnr)=O(mnr)$.

To summarize, we have the following:
\begin{thm} \label{thm:RBP-runtime}
   Suppose that the input rank $r$ $m\times n$ matrix $A$ to the RBP algorithm is $k$--stable for some $k\in\{0,1,\ldots,n-r\}$, i.e.~$A\in\mathcal{M}^{m\times n}(k,r)$.  Let $\delta\in(0,1)$ be given.  Then, with probability at least $1-r\delta$, the total number of operations performed by the RBP algorithm is bounded by $O((k+1)^{-1}mnr^2\log(1/\delta)+mnr)$.
\end{thm}
Note that Theorem \ref{thm:RBP-runtime} only gives a {\it probabilistic} bound on the runtime.  However, the bound can be made deterministic by suitably modifying the RBP algorithm. Specifically, we can add a counter to keep track of the total number of times Step 2b is being executed.  Once that number exceeds $(k+1)^{-1}nr(1+\ln(1/\delta))$, we stop the algorithm and declare failure.  With such modification, the conclusion of Theorem \ref{thm:RBP-sample} still holds.  However, we now have a {\it deterministic} bound of $O((k+1)^{-1}mnr^2\log(1/\delta)+mnr)$ on the runtime.  We remark that such an idea can also be used to develop a ``rank--free'' version of the RBP algorithm, i.e.~one that does not require the knowledge of the rank of the input matrix.  We refer the reader to Section \ref{subsec:rank-free} for details.

The time bound obtained in Theorem \ref{thm:RBP-runtime} compares very favorably with that for the SDP--based algorithm of Cand\`{e}s and Recht \cite{CR09}.  Perhaps more importantly, our algorithm will produce an {\it exact} reconstruction of the input matrix in polynomial time.  By contrast, the Cand\`{e}s--Recht algorithm can only produce an {\it approximate} reconstruction in polynomial time.  This is due to the fact that SDPs can only be solved to a fixed level of accuracy in polynomial time.  We refer the reader to \cite{PK97} for further discussion on this issue.

\subsection{A Rank--Free RBP Algorithm} \label{subsec:rank-free}
Recall that the RBP algorithm introduced earlier assumes that the rank of the input matrix is known.  However, in practice, there is very little a priori information on the input matrix.  This raises the question of whether one can design a reconstruction algorithm that does not need the rank information.  It turns out that this is possible if we modify the RBP algorithm using the idea mentioned at the end of the last sub--section.  Specifically, we keep track of the number of attempts made by the algorithm to find the next basis column.  If that number reaches a pre--specified threshold, say $\Lambda$, then we exit the basis pursuit step and proceed to the row identification step of the algorithm.  The idea is that if $\Lambda$ is sufficiently large and the algorithm fails to find a new basis column after $\Lambda$ drawings, then it probably has found all the basis columns and hence the input matrix can be exactly reconstructed.  To formalize this idea, let us first give a precise description of the proposed algorithm.

\medskip
\noindent{\underline{\sc Rank--Free Randomized Basis Pursuit (RF--RBP) Algorithm}}

\smallskip
\noindent{\,\,\,}\underline{Input}: An $m\times n$ matrix $A$, stopping threshold $\Lambda\ge1$.
\begin{enumerate}
   \item \underline{Initialization}: Initialize $S\leftarrow\emptyset$, $T\leftarrow\{1,\ldots,n\}$ and $\kappa \leftarrow 0$.  The set $S$ will be used to store the column indices that correspond to the recovered basis columns of $A$.  The counter $\kappa$ will be used to keep track of the number of attempts made to find the next basis column.

   \item \underline{Basis Pursuit Step}:
   \begin{enumerate}
      \item If $T=\emptyset$, then stop.  All the columns of $A$ have been examined, and hence $A$ can be reconstructed directly.  Otherwise, reset $\kappa \leftarrow 0$ and proceed to Step 2b.

      \item Let $j$ be drawn from $T$ uniformly at random, and let $u_j\in\R^m$ be the corresponding column of $A$.  Examine all the entries in $u_j$.  

      \medskip
      \noindent If $u_j$ is spanned by the columns whose indices belong to $S$, then increment $\kappa \leftarrow \kappa+1$.  If $\kappa\ge\Lambda$, then proceed to Step 3.  Otherwise, repeat Step 2b.

      \medskip
      \noindent If $u_j$ is not spanned by the columns whose indices belong to $S$, then $u_j$ is a new basis column.  Update:
      $$ S\leftarrow S\cup\{j\},\quad T\leftarrow T\backslash\{j\} $$
      and repeat Step 2.
   \end{enumerate}

   \item \underline{Row Identification}: Let $A_S$ be the $m\times|S|$ sub--matrix of $A$ whose columns are those indexed by $S$.  Find $|S|$ linearly independent rows in $A_S$.  Let $\bar{S}$ be the corresponding set of row indices, and let $A_{\bar{S},S}$ be the corresponding $|S|\times|S|$ matrix.

   \item \underline{Reconstruction}: Examine all the entries in the $i$--th row of $A$ for all $i\in\bar{S}$.  Now, express the $j$--th column of $A$ (where $j\not\in S$) as a linear combination of the basis columns indexed by $S$, where the coefficients are obtained by expressing the vector $(a_{ij})_{i\in\bar{S}}\in\R^{|\bar{S}|}$ as a linear combination of the columns of $A_{\bar{S},S}$.
\end{enumerate}
Again, we are interested in the sampling complexity of the RF--RBP algorithm.  It turns out that if the input matrix is known to be $k$--stable for some $k\ge0$, then the sampling complexity of the RF--RBP algorithm is comparable to that of the RBP algorithm.  Specifically, we prove the following:
\begin{thm} \label{thm:RF-RBP-sample}
   Suppose that the input $m\times n$ matrix $A$ to the RF--RBP algorithm is $k$--stable for some $k\ge k_0$, i.e.~$A\in\mathcal{M}^{m\times n}(k)$.  Let $\delta\in(0,1)$ be given, and set:
\begin{equation} \label{eq:lambda}
   \Lambda = \log\left(\frac{\delta}{\min\{m,n\}}\right) \Big/ \log\left(1-\frac{k_0+1}{n}\right) 
\end{equation}
Then, with probability at least $1-\delta$, the RF--RBP algorithm will terminate with an exact reconstruction of $A$, and the total number of entries inspected by the algorithm is bounded above by $nr+m(r+1)\Lambda$, where $r=\mbox{rank}(A)$.
\end{thm}
{\bf Remarks}
\begin{enumerate}
   \item Since $\log(1-(k_0+1)/n)\le-(k_0+1)/n$, Theorem \ref{thm:RF-RBP-sample} guarantees that the total number of entries inspected by the RF--RBP algorithm is bounded by:
$$ nr + \frac{mn(r+1)}{k_0+1}\log\left(\frac{\min\{m,n\}}{\delta}\right) $$
In particular, when $\delta$ is inversely proportional to a polynomial in $\min\{m,n\}$, the bound above is of the same order as that obtained for the RBP algorithm (see Theorem \ref{thm:RBP-sample}).

   \item To appreciate the power of Theorem \ref{thm:RF-RBP-sample}, consider an $m\times n$ matrix $A$ whose rank $r$ is known to be much smaller than $\min\{m,n\}$, say, $r\le n/2$.  If $A$ is generic, then by Theorem \ref{thm:stab-typ}, it is $k$--stable, where $k=n-r\ge n/2$.  Hence, by Theorem \ref{thm:RF-RBP-sample}, the matrix $A$ can be exactly reconstructed by the RF--RBP algorithm with high probability, and the number of inspected entries is bounded by $O(nr+mr\log n)$.  Note that such a reconstruction is done without the algorithm knowing the exact value of $r$ or $k$.  By contrast, the algorithm of Keshavan et al.~\cite{KMO09} is much less flexible, as it needs to know the exact value of $r$ in order to guarantee an exact reconstruction.
\end{enumerate}
{\bf Proof of Theorem \ref{thm:RF-RBP-sample}}{\quad} For $j=1,2,\ldots,r$, let $q_j$ be the probability that the RF--RBP algorithm finds at least $j$ basis columns before proceeding to Step 3.  We claim that:
   \begin{equation} \label{eq:RF-RBP-basis-prob}
      q_j \ge \prod_{i=1}^j \left[1-\left(1-\frac{k+1}{n-i+1}\right)^\Lambda\right] \qquad\mbox{for }j=1,\ldots,r
   \end{equation}
The proof of (\ref{eq:RF-RBP-basis-prob}) will proceed by induction on $j$.  To facilitate the proof, let us again divide the execution of Step 2 into epochs, where the $(j-1)$--st epoch (for $j=1,2,\ldots,r$) is defined in exactly the same way as in the proof of Theorem \ref{thm:RBP-sample}.  Furthermore, let $p_j$ be the probability that the column selected in an iteration of the $(j-1)$--st epoch is a basis column.  Since $A$ is assumed to be $k$--stable, we have:
$$ p_j \ge \frac{k+1}{n-j+1} \qquad\mbox{for } i=1,2,\ldots,r $$
Now, for $j=1$, we have:
$$ q_1 = 1-(1-p_1)^\Lambda \ge 1-\left(1-\frac{k+1}{n}\right)^\Lambda $$
and hence the base case holds.  Suppose that (\ref{eq:RF-RBP-basis-prob}) holds for some $j<r$.  Then, conditioned on the event that the RF--RBP algorithm finds at least $j$ basis columns before proceeding to Step 3, the probability that the RF--RBP algorithm finds at least $j+1$ basis columns before proceeding to Step 3 is given by:
$$ q_{j+1}^{cond} = 1-(1-p_{j+1})^\Lambda \ge 1-\left(1-\frac{k+1}{n-j}\right)^{\Lambda} $$
Hence, it follows from the definition of conditional probability and the inductive hypothesis that:
$$ q_{j+1} = q_{j+1}^{cond}\cdot q_j \ge \prod_{i=1}^{j+1} \left[1-\left(1-\frac{k+1}{n-i+1}\right)^\Lambda\right] $$
This completes the proof of (\ref{eq:RF-RBP-basis-prob}).

Now, observe that the RF--RBP algorithm will terminate with an exact reconstruction of $A$ iff it finds $r$ basis columns before proceeding to Step 3.  Using (\ref{eq:RF-RBP-basis-prob}) and the definition of $\Lambda$ in (\ref{eq:lambda}), we see that the probability of such an event is at least:
\begin{eqnarray*}
   \prod_{i=1}^r \left[1-\left(1-\frac{k+1}{n-i+1}\right)^\Lambda\right] &\ge& \left[1-\left(1-\frac{k+1}{n}\right)^\Lambda\right]^r \\
   \noalign{\medskip}
   &\ge& \left(1-\frac{\delta}{\min\{m,n\}}\right)^{\min\{m,n\}} \\
   \noalign{\medskip}
   &\ge& 1-\delta
\end{eqnarray*}
Moreover, the number of distinct columns inspected by the algorithm is always bounded above by $(r+1)\Lambda$, which implies that the total number of entries inspected by the algorithm is bounded above by $nr+m(r+1)\Lambda$.  This completes the proof of Theorem \ref{thm:RF-RBP-sample}. \endproof

\medskip
\noindent Finally, upon following the proof of Theorem \ref{thm:RBP-runtime}, one can easily establish the following complexity result for the RF--RBP algorithm:
\begin{thm} \label{thm:RF-RBP-runtime}
   Given an $m\times n$ matrix $A$ and a stopping threshold $\Lambda\ge1$, the total number of operations performed by the RF--RBP algorithm before it terminates is bounded by $O(mr^2\Lambda + mnr)$, where $r=\mbox{rank}(A)$.
\end{thm}
We remark that the bound in Theorem \ref{thm:RF-RBP-runtime} holds for {\it arbitrary} input matrices.  In the case where the input matrix has rank $r$ and is $k$--stable, we can set $\Lambda$ as in (\ref{eq:lambda}) and bound the total number of operations by:
$$ O\left(\frac{mnr^2}{k+1}\log\left(\frac{\min\{m,n\}}{\delta}\right)+mnr\right) $$
In particular, when $\delta$ is inversely proportional to a polynomial in $\min\{m,n\}$, the bound above is of the same order as that obtained for the RBP algorithm (see Theorem \ref{thm:RBP-runtime}).

\section{Conclusion} \label{sec:concl}
In this paper we proposed a randomized basis pursuit (RBP) algorithm for the matrix reconstruction problem.  We introduced the notion of a $k$--stable matrix and showed that the RBP algorithm can reconstruct a $k$--stable rank $r$ $n\times n$ matrix with high probability after inspecting $O((k+1)^{-1}n^2r\log n)$ of its entries.  In addition, we showed that the runtime of the RBP algorithm is bounded by $O((k+1)^{-1}n^2r^2\log n+n^2r)$.  Our results yield substantial improvement over those in existing literature (\cite{CR09,KMO09,CT09}), in the sense that the RBP algorithm can reconstruct a {\it larger class of matrices} by inspecting a {\it smaller number of entries}, and it can do so in a {\it more efficient} manner.  Although the RBP algorithm assumes that the rank of the input matrix is known, we showed that such an assumption can be removed.  Specifically, we proposed a variant of the RBP algorithm that can reconstruct a matrix without knowing the exact value of its rank.  Such a feature offers great flexibility in practical settings.  Finally, it would be interesting to study the tradeoff between the sampling complexity and computational complexity of the matrix reconstruction problem.  Another interesting direction would be to extend our techniques to handle the case where the sampled entries are {\it noisy}.  Some recent results along this direction, which are established using the techniques of \cite{CR09,CT09}, can be found in \cite{CP09}.

\bibliography{sdpbib}
\bibliographystyle{abbrv}

\end{document}